\documentclass[11 pt]{article}

\usepackage{comment}

\usepackage{amsmath,amsfonts,amssymb}
\newtheorem{theorem}{Theorem}[section]

\newtheorem{nt}{Selfnote}

\newtheorem{definition}[theorem]{Definition}

\newtheorem{lemma}[theorem]{Lemma}
\newtheorem{proposition}[theorem]{Proposition}
\newtheorem{corollary}[theorem]{Corollary}
\newtheorem{claim}[theorem]{Claim}
\newtheorem{fact}[theorem]{Fact}

\newtheorem{remk}[theorem]{Remark}
\newtheorem{exmp}[theorem]{Example}

\newenvironment{remark}{\begin{remk}
\begin{normalfont}}{\end{normalfont}
\end{remk}}

\def\FullBox{\hbox{\vrule width 8pt height 8pt depth 0pt}}

\def\qed{\ifmmode\qquad\FullBox\else{\unskip\nobreak\hfil
\penalty50\hskip1em\null\nobreak\hfil\FullBox
\parfillskip=0pt\finalhyphendemerits=0\endgraf}\fi}

\def\qedsketch{\ifmmode\Box\else{\unskip\nobreak\hfil
\penalty50\hskip1em\null\nobreak\hfil$\Box$
\parfillskip=0pt\finalhyphendemerits=0\endgraf}\fi}

\newenvironment{proof}{\begin{trivlist} \item {\bf Proof:~~}}
  {\qed\end{trivlist}}

\newcommand{\etal}{{\it et~al.\ }}
\newcommand{\ie} {{\it i.e.,\ }}

\newfont{\inhead}{eufm10 scaled\magstep1}
\def\abs#1{\left| #1 \right|}

\newcommand{\N}{{\mathbb{N}}}

\newcommand{\poly}{{\mathrm{poly}}}

\renewcommand{\Pr}{\mathop{\mathbb P}\displaylimits}

\newcommand{\eps}{\varepsilon}

\newcommand{\Exs}[2]{\mathbb{E}_{#1}\left[ #2 \right]}

\newcommand{\mcal}[1]{\mathcal{#1}}

\newcommand{\bra}[1]{\langle #1|}
\newcommand{\ket}[1]{|#1\rangle}

\newcommand{\Tr}{\mbox{\rm Tr}}

\newtheorem{observation}[theorem]{Observation}

\newcommand{\tnote}[1]{[{\bf Thomas: #1]}}

\begin{document}

\title{Near-optimal extractors against quantum storage}
\author{Anindya De\thanks{Computer Science Division, University of California, Berkeley, CA, USA. \texttt{anindya@cs.berkeley.edu}. 
Supported by the ``Berkeley Fellowship for Graduate Study'' and Luca Trevisan's  US-Israel BSF grant 2006060.} \and Thomas Vidick\thanks{Computer Science Division, University of California, Berkeley, CA, USA. \texttt{vidick@cs.berkeley.edu}. Supported by ARO grant W911NF-09-1-0440 and NSF grant CCF-0905626.}
}
%

\begin{titlepage}

\maketitle

\begin{abstract}
We show that  Trevisan's extractor and its variants~\cite{T99,RRV99} are secure against bounded quantum storage adversaries. 
 One instantiation gives the first such extractor to achieve an output length $\Theta(K-b)$, where $K$ is the source's entropy and $b$ the adversary's storage,  together with a poly-logarithmic seed length. Another instantiation achieves a logarithmic key length, with a slightly smaller output length $\Theta((K-b)/K^\gamma)$ for any $\gamma>0$. In contrast, the previous best construction~\cite{TS09} could only extract $(K/b)^{1/15}$ bits. 
Some of our constructions have the additional advantage that every bit of the output is a function of only a polylogarithmic number of bits from the source, which is crucial for some cryptographic applications. 

 Our argument is based on bounds for a generalization of quantum random access codes, which we call \emph{quantum functional access codes}. This is crucial as it lets us avoid the local list-decoding algorithm central to the approach in~\cite{TS09}, which was the source of the multiplicative overhead.
\end{abstract} 

\thispagestyle{empty} 
\vfill
 
\end{titlepage}
 
\section{Introduction}\label{sec:intro}

Randomness extractors are fundamental building blocks in pseudorandomness theory, with many applications to derandomization, error-correcting codes, and expanders, among others.
They are also of central importance in cryptography, where they are often used to build key generation primitives. In this context, one usually has the notion of an adversary, a malicious observer who is trying to discover a bit of the honest player's output. A prominent model for adversaries is the bounded storage model, introduced by Maurer~\cite{Maurer92}, in which the adversary is allowed to store a limited amount of information about the extractor's input. 

Formally, we say that a function $\text{Ext}:\{0,1\}^N \times \{0,1\}^t \rightarrow \{0,1\}^m$ is  a $(K,\eps)$ strong extractor if for every distribution $X$ with  min-entropy  at least $K$ ($X$ is called the \emph{source}) and uniformly random $Y$ (called the \emph{seed}), the distribution $(Y,\text{Ext}(X,Y))$ is within a statistical distance of at most $\eps$ from the uniform distribution. The extractor is said to be secure against $b$ bits of storage if $Ext(X,Y)$ is $\eps$-close to uniform even from the point of view of an adversary who has been allowed to store $b$ bits of information about $X$, and has also later been revealed the seed $Y$.

Constructions of extractors are known that are almost-optimal in all parameters, even in the presence of the adversary (in fact, a result by Lu~\cite{L04} shows that any $(K,\eps)$ strong extractor is essentially a $(K+b,\eps)$ extractor secure against $b$ bits of storage). Nevertheless, in a world in which no adversary can be trusted, K\"onig~\etal~\cite{KMR05} introduced the following interesting twist: what if the adversary is allowed \emph{quantum} memory? In this setting, the fundamental difficulty that arises is a familiar one, with a long history: how much information can be encoded in a quantum state? 

The fact that this question can admit very different answers depending on its precise formulation is reflected in the fact that some, but not all,  classical extractor constructions are secure in the presence of a quantum adversary, as was demonstrated in~\cite{GKKRW}. While many constructions have been shown to be sound on a case-by-case basis~\cite{KMR05,KT07,FS08,TS09}, all have parameters that are far from optimal either in terms of seed length or of output length.

Central to the proof of our result are bounds on a construct which we call quantum functional access codes (QFAC), and we introduce them next.

\paragraph{Quantum functional access codes.} Holevo~\cite{Hol73} was the first to tackle the question of the information capacity of a quantum state, showing that one needs at least $n$ qubits in order to encode $n$ bits of information. However, this bound only holds when it is required that the whole $n$ bits be recoverable from the quantum storage. As such, it is generally not applicable in a cryptographic context, where typically even partial information is important. Instead of asking for the whole input $x\in\{0,1\}^n$ to be recoverable from its encoding $\Psi(x)$, Ambainis~\etal~\cite{ANTV02} consider encodings in which it is only required that any bit of $x$ can be recovered from $\Psi(x)$ with probability $1/2+\eps$ (over the measurement's randomness), and they call such encodings `random access codes' (RACs). Note that, since the encoding is quantum, the recoverability of any one bit does not imply the recoverability of the whole string $x$, so that Holevo's bound does not apply. Nevertheless, Ambainis~\etal showed that RACs require essentially $(1-H(1/2+\eps))n$ qubits to encode $n$ bits, where $H$ is the binary entropy function, providing a linear lower bound for fixed $\eps$. These bounds have proved instrumental in many results in information theory. In fact, as pointed out in~\cite{TS09}, random access codes provide a way to construct one-bit extractors that are secure against quantum storage.

We push this question even further: what if, instead of asking that the encoding lets us recover any bit of the input, we asked that it lets us recover any one out of some fixed set of functions of the input? For example, we could ask about encodings that let us recover the XOR of any $k$ bits of the input\footnote{Such codes were introduced in~\cite{ARW08}, where they are called XOR-QRACs}, but one can also consider more general settings. 

One might ask about the relevance of such encodings, when we already know that there are strong linear lower bounds on RACs --- surely, these will extend to any encoding which lets us recover more than any single bit of the input. The key point here is that, even though both Holevo's bound and the RAC lower bounds are linear in the input length when the success probability $p$ is fixed, the two bounds scale very differently when one considers the dependence on $p$: while an improvement on Holevo's lower bound, due to Nayak and Salzman~\cite{NS06}, scales as $n-\log 1/p$, the RAC bound scales as $(4\eps^2/\ln 2)n$ for small $\eps = 2p-1$. So we are asking, how does the minimal length of the code scale with the success probability, depending on the set of functions that we are trying to recover?

Define a $(n,b,\eps)$ QFAC for a set of $n$-bit strings $A$ and a set of functions $\mathcal{C}$ from $A$ to $\{0,1\}$ as a $b$-qubit encoding of strings $x\in A$ such that, for any function $f\in\mathcal{C}$, one can recover $f(x)$ from the encoding of $x$ with success probability $1/2+\eps.$\footnote{A RAC is then simply a QFAC for the set of coordinate functions $f_i:x\mapsto x_i$.} Intuitively, the more the set of functions $\mathcal{C}$ is error resilient (i.e., the more spread-out the images $(f(x))_{f\in\mathcal{C}} \in \{0,1\}^{|\mathcal{C}|}$), the stronger the lower bound should be on the length of the encoding.
For example, using a simple reduction to known results we can show that any $(n,b,\eps)$ QFAC for the set $\mathcal{C} = \{f_y:x\mapsto x\cdot y\mod 2,\, y\in\{0,1\}^n\}$ must have length $b\geq n-\log 1/\eps$. If one simply used the fact that any bit of $x$ can be recovered from such a QFAC with probability $1/2+\eps$, the resulting bound would be the much weaker $O(\eps^2 n)$.

We believe that QFACs constitute a primitive that should be of wide interest in studying the properties of quantum states from an information-theoretic point of view. In this paper, we demonstrate the relevance of this construct by showing how good bounds on some QFACs can be used to prove the security of an extractor against quantum storage with almost-optimal parameters. In fact, many previous constructions of extractors against quantum storage can be seen as implicitly proving bounds on QFACs. For example, the construction in~\cite{KMR05} shows that any $(n,b,\eps)$ QFAC for a set of $2$-universal hashing functions must have length $b\geq n-2\log 1/2\eps$. 

\paragraph{Techniques.} In this section we give an overview of our proof technique, explaining the connection between extractors and QFACs in the context of Trevisan's general construction paradigm~\cite{T99}. To describe this, let us first give a brief overview of the main steps that go into the proof of the construction by Ta-Shma~\cite{TS09}. 

The construction starts by encoding the weakly random source $x \sim X$  using a locally list-decodable code $C$~\cite{STV99}. This is  followed by an application of the Nisan-Wigderson generator \cite{NW94}, interpreting $C(x)$ as the truth table of the ``hard'' function. 

The proof of correctness for this construction, as the first part of ours, follows the general reconstruction framework of~\cite{T99}. For the sake of contradiction, assume that there is a test $T$, which performs a measurement on the adversary's quantum encoding $\Psi(x)$ in order to distinguish the output from uniform with advantage $\eps$. A Markov argument shows that for at least an $\eps/2$ fraction of the samplings $x$ from the source (call them bad samplings), $T$ can distinguish the output (when the source is $x$) from uniform with success at least $\eps/2$. Consider any such bad sampling $x$. A standard hybrid argument, along with properties of the Nisan-Wigderson generator, allows us to construct a circuit $T'$ (using little non-uniformity about $x$) which predicts a random position of $C(x)$ with probability $\frac{1}{2}+\frac{\eps}{m}$. Further, $T'$ makes exactly one query to $T$. 

At this stage, we have constructed a small circuit $T'$, which uses the adversary's quantum information in order to predict the bits of $C(x)$ with some small success probability. The proof in~\cite{TS09} shows how from such a circuit, one can construct another circuit which predicts any position of $x$ with probability $0.99$ and queries $T'$ at most  $q=(m/\eps)^c$ times ($c=15$ for the code in~\cite{TS09}). This gives a random access code for $x$; however since it makes $q$ measurements on the quantum state $\Psi(x)$, the no-cloning theorem forces us to see it as having a length of $q\cdot b$ qubits. The main drawback of this method is  that the quantum state needs to be copied a large number of times in order to get a RAC -- thus yielding a weaker bound on the output length than one might hope for. 

Our proof departs fundamentally from the usual reconstruction paradigm at this point: instead of using a short RAC for $C(x)$ to construct a longer RAC for $x$, we give a direct analytical argument showing that any RAC for $C(x)$ must be long. Note that a RAC for $C(x)$ is simply a QFAC for the class of functions $f_i:x\mapsto C(x)_i$. Intuitively, such a QFAC cannot be short, even though its success probability $1/2+\eps/m$ is small. If the QFAC is classical ({\em C}FAC), this is easy to show: assume that there existed a short CFAC for this problem. One can just repeat the recovery procedure to get a string $y$ that agrees with $C(x)$ at a fraction $1/2+\eps/m$ of positions, and then one can use the good list-decoding properties of $C$ to argue that the CFAC essentially lets us recover the whole input $x$, and hence must be long.  In the quantum setting, however, it is far from obvious if this is true, the primary difficulty being that we cannot repeat the recovery procedure, since it involves measuring a quantum state.

In order to overcome this difficulty, we directly prove a lower bound on the length of the QFAC derived from the code. This lets us derive a contradiction, proving that our extractor is safe against quantum storage. The idea for the lower bound consists in seeing any good QFAC as an adversary which uses small memory, and is able to predict codeword positions.  Using the fact that list decodable codes can be interpreted as one-bit strong extractors, we can then use a a result by Koenig and Terhal~\cite{KT07} to show that such an adversary would imply a classical adversary with similar storage against a one-bit strong extractor, which we know does not exist. This leads to a contradiction, thus proving the lower bound.

\paragraph{Our results}
We show that any extractor based on Trevisan's construction paradigm and its variants~\cite{T99,RRV99} is also safe against a bounded quantum storage adversary, with almost the same parameters as the classical construction. Rather than give the full technical result here (see Theorem~\ref{thm:XORextmain}), we discuss instantiations with two specific codes. 

We first use a code from~\cite{GHSZ00}, which is obtained through the concatenation of the Reed-Solomon code and the Hadamard code. This lets us prove the following: 

\begin{theorem}\label{thm:main} For any constants $\gamma,c,c'>0$, there is a polynomial-time computable function $Ext:\{0,1\}^{N} \times \{0,1\}^t \rightarrow \{0,1\}^m$, where $t = O(\log N)$ and $m=\Omega\left(\frac{K-b}{K^\gamma}\right)$, which is a $(K, N^{-c})$ extractor against $b$ qubits of quantum storage, for any $K\geq N^{c'}$. 
\end{theorem}

We note that the construction in~\cite{TS09} uses the concatenation of a Reed-Muller code with the Hadamard code, the parameters of the Reed-Muller code being chosen so that one can do local list-decoding. In contrast, our analysis just needs a good list-decoding radius, but no local list-decoding property.  Hence our result carries over to~\cite{TS09} and in particular implies that the construction in~\cite{TS09} has much better output length than the one shown in that paper, which was $\Omega((K/b)^{1/15})$.

This first construction does not have the desirable property of \emph{local computability}. By using a different code, we can also show the following:

\begin{theorem}
For any constants $\alpha,c>0$, there is a function $Ext:\{0,1\}^{N} \times \{0,1\}^t \rightarrow \{0,1\}^m$, where $t = O(\log^4 N)$ and $m=\Omega(\alpha N - b )$, which is a $(\alpha N, N^{-c})$ extractor against $b$ qubits of quantum storage. Moreover, each bit output by the extractor is computable in $\poly \log N$ time. 
\end{theorem}

Even though it has a slightly larger seed length (note however that its output length is the optimal $\Theta(\alpha N - b)$), a major advantage of this extractor is its simplicity: each bit of the output is simply the XOR of $O(\log N)$ bits of the source, chosen based on the seed. In particular, it is locally computable\footnote{For this to hold, we also need to check that the bits to be XOR-ed can be chosen in poly-logarithmic time, which is the case in this construction.}. On the other hand, that construction is restricted to extracting from linear entropy rates. This is inevitable, as lower bounds by Viola~\cite{V03} show that locally computable extractors cannot extract from sources with entropy less than $N^{0.99}$ using a polylogarithmic seed length. 

The QFACs at the heart of this second construction are in fact the XOR-QRACs from~\cite{ARW08}. A by-product of our proof is an improvement of the lower bound proved in that paper on the length of such codes (see Corollary~\ref{cor:xorqfac}).

\medskip

A nice side feature of both these constructions, especially if one is interested in cryptographic applications, is that it is possible to achieve an arbitrary inverse polynomial statistical distance from the uniform distribution, while paying only a polylogarithmic cost in terms of output length and seed length (this will be apparent from the more detailed statement of Theorem~\ref{thm:main} given in Section~\ref{sec:main}). This property was not known to hold for previous short seed extractor constructions against quantum storage.

\paragraph{Applications to cryptography.}
Our results are of direct applicability to the following key expansion scenario.
 Alice and Bob share a small secret uniformly random key $k$. They would like to expand it into a longer key $k'$ in order to securely communicate in presence of an adversary Eve. A public source of weak randomness $R$ (assume that $R$ has min-entropy at least $K$) is available to all parties. When the string $R$ is broadcast, Eve is allowed to compute an arbitrary function $\Psi$ which maps the input  to a state on $b$ qubits \ie $\Psi:\{0,1\}^{|R|} \rightarrow \mathbb{C}^{2^b \times 2^b}$ and store the result. However, once she stores $\Psi(R)$, her access to $R$ is cut off. The goal is to come up with an efficient function $Ext$ which can be used by Alice and Bob to compute the shared string $k'=Ext(R,k)$. The required security condition is that $k'$ is close to being uniformly random to Eve, even given her knowledge of $\Psi(R)$. In fact, we would like $k'$ to remain random even if $k$ is later revealed to Eve (after $\Psi(R)$ is computed and access to $R$ has been cut off). 

For this application, it is important that $Ext$ be \emph{locally computable}, i.e. individual bits of the output should be a function of a polylogarithmic number of bits of the source $R$. Indeed, since we are putting a cap on the adversary's storage it would be unreasonable not to put a similar cap on the memory used by the honest parties Alice and Bob to compute bits of their shared key. 

Our second construction has the property of being locally computable: every bit of the output is a function of polylogarithmically many bits from the source. While various constructions of classical locally computable extractors are already known~\cite{DM04,L04,V04,DT09}, ours are the first to be proved secure against quantum adversaries. This makes them particularly suitable for use in the context of bounded storage cryptography. 

We note here that the results in this paper have recently been extended by Portmann, Renner, and the authors to show that Trevisan's extractor is secure in a broader context than that of the bounded-storage model~\cite{DPRV09}: they show security when one has a lower bound on the \emph{conditional min-entropy} of the source, conditioned on the adversary's quantum information. This is a more general assumption since it is implied by the bounded storage assumption, but the converse is not true in general. Proving security in this setting is crucial in a cryptographic context, as it allows secure composability of the extractor with other cryptographic primitives.

\paragraph{Organization of the paper.} We start with some preliminaries in Section~\ref{sec:prelim}. In Section~\ref{sec:qfac} we introduce quantum functional access codes and give bounds for some specific families of these codes. In Section~\ref{sec:main} we describe our construction and state its parameters. Finally, the proof of security is given in Section~\ref{sec:mainsecurity}.
 

\section{Notation and Preliminaries}\label{sec:prelim}
The following notations are used throughout the paper.  For $x \in \{0,1\}^n$, $x_i$ denotes the $i^{th}$ bit of $x$. Given two $n$-bit strings $x,y$, we let $\Delta(x,y)$ denote their relative Hamming distance, i.e. the fraction of positions at which they differ.  $\mathcal{D}_b$ denotes the set of all density matrices on $b$ qubits (complex $2^b$-dimensional positive matrices with trace $1$). In general, a measurement $M$ is described by a list of positive operators $M_a$ such that $\sum_a M_a^\dagger M_a = Id$. The probability that outcome $a$ is observed when the measurement is performed on a density matrix $\rho$ is then given by $\Tr(M_a\rho M_a^\dagger)$. 
All logarithms are taken in base $2$. Throughout, $H$ will denote the binary entropy function $H(x) = - x\log x - (1-x)\log(1-x)$ for $0<x<1$. We set the convention that $H(0)=H(1)=0$.

\paragraph{Distributions.} The uniform distribution on $\{0,1\}^n$ is denoted by $U_n$. 
We will manipulate random variables that have both classical and quantum parts. In general, given two classical random variables $X$, $Y$, $X\circ Y$ is the same as the random variable $(X,Y)$. Given two states $\rho,\sigma$, $\rho\circ \sigma$ is just $\rho\otimes \sigma$. Finally, given a classical random variable $X:\Omega \rightarrow \{0,1\}^n$ and a quantum random variable $\rho:\Omega\rightarrow \mathcal{D}_b$, $X\circ \rho$ denotes the state $\Exs{w\in\Omega}{ \ket{X(w)}\bra{X(w)}\otimes \rho(w)}$. The statistical distance between two distributions $D_1$ and $D_2$ (or, more generally, the trace distance when these distributions involve quantum components) is denoted by $\|D_1-D_2\|$.

\begin{definition}
A (classical) distribution $X$ is said to have min-entropy at least $K$ (denoted $H_{\infty}(X) \geq K$) if  $\forall x$, $Pr[X=x] \leq 2^{-K}$. 
\end{definition}

\paragraph{Extractors.} We first give the the formal definition of a strong extractor.

\begin{definition}\label{def:strongext}
$Ext:\{0,1\}^N \times \{0,1\}^t \rightarrow \{0,1\}^m$ is said to be a $(K,\eps)$ strong extractor if for every distribution $X$ with min-entropy at least $K$, we have that $\|U_{m+t} - Ext(X,U_t) \circ U_t\| \leq \eps$. Here both $U_t$'s in the second expression correspond to the same sampling.

 $X$ is usually called the \emph{source} (and $N$ its length), while the extractor's second input is called the \emph{seed} (of length $t$). 
\end{definition}

We now extend this definition to that of a strong extractor secure against a bounded-storage quantum adversary.

\begin{definition}\label{def:strongquantum}
$Ext:\{0,1\}^N \times \{0,1\}^t \rightarrow \{0,1\}^m$ is said to be a $(K,\eps)$ strong extractor against $b$ qubits of quantum storage if for every map $\Psi:\{0,1\}^N \rightarrow \mathcal{D}_{b}$ and every distribution $X$ such that $H_{\infty}(X) \geq K$
\begin{align}\label{eq:stat}
\|U_m \circ \Psi(X) \circ U_t - Ext(X,U_t) \circ \Psi(X) \circ U_t\| \leq \eps
\end{align}
where both $U_t$'s in the second expression correspond to the same sampling. 
\end{definition}

We note that condition (\ref{eq:stat}) above is equivalent to requiring that for any collection of measurements $\{M_{u,y}^0,M_{u,y}^1\}$ on $\mathcal{D}_b$,
\begin{align*} 
\Big|\mathbb{E}_{x\sim X, y\sim U_t}\Big[\text{Tr}\left(\Exs{u\in\{0,1\}^m}{M_{u,y}^1 \Psi(x)(M_{u,y}^1)^\dagger} \right)
-\Tr\left(M_{Ext(x,y),y}^1 \Psi(x)(M_{Ext(x,y),y}^1)^\dagger\right)\Big]\Big| \leq \eps 
\end{align*}

\paragraph{Quantum codes.}
A $(n,b)$ quantum encoding is a map $\Psi: \{0,1\}^n \rightarrow \mathcal{D}_b$. A fundamental theorem due to Holevo states that, for any fixed measurement $M$, the outcome of that measurement when performed on $\Psi(x)$ cannot contain more information about $x$ than a classical string of $b$ bits:

\begin{theorem} \cite{Hol73} \label{thm:holevo} 
Let $X$ be any distribution on $\{0,1\}^n$ and $\Psi(X)=\mathbb{E}_{x \in X} [\Psi(x)]$.  For a particular measurement $M$, let $Y_M$ denote the classical random variable resulting from applying the measurement on $\Psi(X)$. If $I(X:Y)$ denotes the mutual information of $X$ and $Y$ and $S(\Psi(X))$ denotes the von Neumann entropy of $\Psi(X)$, then $I(X:Y) \leq S(\Psi(X))$. 
\end{theorem}

\paragraph{Oracle circuits.} Our proofs of security will involve the construction of oracle circuits. If $A$ is an oracle circuit, we denote by  $A^B$ the circuit that uses $B$ as the oracle. Further, let $C$ be an oracle machine which uses $A$ as an oracle (denoted by $C^A$). Then it is understood that when $C$ calls $A$, then $A$ calls the appropriate oracle $B$. Thus $C^A \equiv C^{A^B}$. We will say that a circuit $C:\{0,1\}^{n+t}\to \{0,1\}$ computes a function $f$ with $t$ bits of advice if  there is a string $a\in\{0,1\}^t$ such that for every $x\in\{0,1\}^n$, $C(x,a)=f(x)$.

We will use the following easy claim:

\begin{claim}\label{compo}
Let $B$ be any oracle such that oracle circuit $A$ can be constructed using at most $t_1$ bits of advice and $A $ queries $B$ at most $q_1$ times. Again let $C$ be an oracle circuit which queries $A$ and $C$ can be constructed using at most $t_2$ bits of advice. Further, $C$ queries $A$ at most $q_2$ times. Then $C$ can be considered as an oracle circuit which queries $B$ at most $q_1q_2$ times and can be constructed using at most $t_1+t_2$ bits of advice. 
\end{claim}


\section{Quantum Functional Access Codes}\label{sec:qfac}

Consider the following problem from the theory of error-correcting codes. Let $C: \{0,1\}^n\rightarrow \{0,1\}^m$ be a code which is $(\eps,L)$ list-decodable i.e. for any $x \in \{0,1\}^m$, there are at most $L$ codewords $y$ such that $\Delta(x,y)\leq \frac{1}{2} - \eps$. Let $A = \{C(x):\ x\in\{0,1\}^n\}$ be the set of all codewords, and consider $Enc:A\rightarrow \{0,1\}^b$, a probabilistic encoding such that for every $z \in A$, $z_i$ can be recovered from $Enc(z)$ with probability $\frac{1}{2}+ 2\eps$, on average over the choice of $i\in[m]$. Given $Enc(z)$, by performing the recovery procedure for every index $i$, we obtain a string $y$ which will agrees with $z$ on at least a $\frac{1}{2} + \eps$ fraction of the positions with high probability. But then the exact element $z$ can be recovered using just an additional $\log |L|$ bits of advice (as per the list-decodability property of $C$). Hence, $Enc$ can be seen as a high-probability encoding of any codeword, using only $b+ \log |L|$ bits. However, the obvious information-theoretic bound shows that this must be at least $\log |C|$ bits, implying that $b \geq \log \frac{|C|}{|L|}$. This is  much better than the bound $b\geq O(\eps^2 \log|C|/\log n)$, for small $\eps$, that one gets if there is no guarantee on the structure of the set $C$ (see Theorem~3.2 in~\cite{TS09} for a proof).

\bigskip

To model this situation more precisely, note that the recovery procedure lets us recover any bit of $C(x)$ with non-trivial probability. As such, $Enc$ can be seen as a probabilistic encoding of every $x\in\{0,1\}^n$ which lets us evaluate a class of functions $\mathcal{C} = \{g_i : x\mapsto C(x)_i,\ i\in [n]\}$. This is a generalization of the usual random access codes, introduced in~\cite{ANTV02}, for which $\mathcal{C} = \mathcal{C}_1 = \{g_i :x\mapsto x_i,\, i\in [n]\}$.

It is natural to expect that lower bounds for this more demanding kind of random access code would be tighter than more general lower bounds, in a way that depends on the structure of $\mathcal{C}$. We introduce the following definition:

\begin{definition}
Let $A \subset \{0,1\}^n$, and $\mathcal{C} \subset \{f:A \rightarrow \{0,1\}\}$ a set of functions defined on $A$. For $\eps \in (0,1/2]$, a $(n,b,\eps)$ quantum functional access code, or QFAC, for $(A,\mathcal{C})$ is a map $\Psi: A \rightarrow \mathcal{D}_b$ such that, for every $f\in\mathcal{C}$, there is a measurement $M_f = \{M_f^0,M_f^1\}$ such that for every $x\in A$, $\Tr(M_f^{f(x)}\Psi(x)(M_f^{f(x)})^\dagger) \geq 1/2+\eps$. If this first property only holds on average over the choice of $f$, then we'll say that $\Psi$ is a $(n,b,\eps)$ QFAC \emph{on average} for $(A,\mathcal{C})$.
\end{definition}

The discussion above shows a strong lower bound on the length of any \emph{classical} functional access code for a set of functions $\mathcal{C}= \{g_i : x\mapsto C(x)_i,\ i\in [n]\}$ that is derived from a good list-decodable code $C$. However, the classical argument cannot be extended in a straightforward manner to the quantum case, as it is dependent upon performing successive measurements on the encoding. If the encoding is quantum, the first such measurement will destroy the state, and we will not be able to proceed further. 

Nevertheless, for some specific cases, such bounds follow from previously known results. We start with the standard setting of random access codes, for which Theorem~4.1 in~\cite{ANTV02} implies the following (see also Theorem~3.2 in~\cite{TS09}):

\begin{lemma}\label{qbound}
Let $A\subseteq \{0,1\}^n$ and $\eps\in(0,1/2]$ such that there exists a $(n,b,\eps)$ quantum functional access code for $(A,\mathcal{C}_1)$. Then $\log |A| \leq O\left(\frac{b\log n}{\eps^2} \right)$.\footnote{As noted in~\cite{TS09}, the loss of a factor $\log n$ is inevitable. Note however that this can be removed in the case where $A=\{0,1\}^n$ by following the proof for quantum random access codes in~\cite{ANTV02}.}
\end{lemma}

Central to this work is the fact that functional access codes for larger classes of functions than the simple coordinate functions $\mathcal{C}_1$ enjoy much stronger lower bounds, with a weaker dependence on the success probability $\eps$. K\"onig, Maurer and Renner~\cite{KMR05} show the following:

\begin{theorem}[\cite{KMR05}, Thm. 12 and Cor. 13] Let $\mathcal{C}$ be the set of all functions from $\{0,1\}^n$ to $\{0,1\}$. Then any $(n,b,\eps)$ QFAC on average for $(A,\mathcal{C})$ satisfies $\log |A| \leq b+2\log 1/2\eps$. Moreover, the same bound holds if $\mathcal{C}$ is any family of two-universal hash functions, and the decoding procedure is only required to be correct on average over the choice of $x\in A$.
\end{theorem}

There is an obvious connection between lower bounds on the length of QFACs and lower bounds on one-way quantum communication complexity, even though results in the latter setting usually do not focus on the error dependence as much as is needed for our applications. Nevertheless, the following bound easily follows from known results:

\begin{lemma} Let $\mathcal{C} = \{g_y:x\mapsto x\cdot y\mod 2,\ y\in\{0,1\}^n\}$. If there exists a $(n,b,\eps)$ QFAC for $(A,\mathcal{C})$, then  $\log |A|\leq b+ 2\log(1/2\eps)$.
\end{lemma}

\begin{proof} Note that any $(n,b,\eps)$ QFAC for $(A,\mathcal{C})$ implies a one-way quantum protocol for the communication problem in which Alice is given $x\in A$, Bob is given $y\in\{0,1\}^n$, and their goal is to output $x\cdot y \mod 2$. Using a reduction from~\cite{CDNT98}, any such protocol communicating $b$ qubits and succeeding with probability $1/2+\eps$ can be transformed into a protocol that sends any $x\in A$ to Bob, using $b$ qubits, with success probability $4\eps^2$. Theorem~1.1 in~\cite{NS06} then shows that $b \geq \log|A| - \log(1/4\eps^2)$.
\end{proof}

Families of two-universal hash functions over $\{0,1\}^n$, as well as the Hadamard code, both have size $\Omega(2^n)$, which makes them unsuitable for our purposes. Indeed, in our applications to extractors we will use the seed to select a few random functions from the family $\mathcal{C}$ and apply them to the source in order to obtain the output. However, using any of the last two function families  would require a seed of length linear in the source length, whereas we would like it to be poly-logarithmic.

Our main result relies on the fact, proved below, that there are no short QFACs for families of functions that are defined from list-decodable codes. This extends the discussion introducing this section to the case of quantum encodings, and in fact we will get essentially the same bound as stated there --- even though, as we argued was necessary, the proof will be very different. It will be useful to consider \emph{approximately} list-decodable codes, which we define as follows:

\begin{definition} Let $\eps,\delta>0$ and $L\in\N$. A code $C:\{0,1\}^N \rightarrow \{0,1\}^{\overline{N}}$ is $(\eps,\delta,L)$ approximately list-decodable if for every $x\in\{0,1\}^{\overline N}$, there exists at most $L$ strings $\{y_i\}_{i=1}^L \in \{0,1\}^N$, such that for any string $z\in\{0,1\}^N$ satisfying $\Delta(x,C(z)) < 1/2-\eps$, $\exists i \in \{1,\ldots,L\}$ such that $\Delta(z,y_i)\leq\delta$. If $C$ is $(\eps,0,L)$ approximately-list decodable then we simply say that $C$ is $(\eps,L)$ list-decodable.
\end{definition}

\begin{proposition}\label{prop:listqfac}
Let $\eps,\delta>0$ and $L\in\N$. Let $C:\{0,1\}^N\to\{0,1\}^{\overline N}$ be a $(\eps/2,\delta,L)$ approximately list-decodable code, and $\mathcal{C} = \{f_i:x\mapsto C(x)_i,\, i\in [\overline N]\}$. Let $A\subseteq \{0,1\}^N$, and suppose that there exists a $(N,b,\eps)$ QFAC on average for $(A,\mathcal{C})$. Then 
$$\log |A| < H(\delta)N+b +\log L +O(\log 1/\eps)$$
Moreover, this bound holds even when we only require the QFAC to have success probability $1/2+\eps$ on average over the choice of $x\in A$, instead of for all $x$. 
\end{proposition}

The proof crucially relies on the result by K\"onig and Terhal~\cite{KT07} that strong one-bit extractors are automatically safe against quantum adversaries, in some range of parameters. It proceeds through the following three steps:
 \begin{enumerate}
 \item Show that any $(\eps,\delta,L)$ approximately list-decodable code $C$ defines a good $1$-bit classical strong extractor.
 \item Use Theorem~III.1 from~\cite{KT07} to show that the previous extractor is automatically safe against quantum adversaries that are allowed a bounded amount of storage. 
 \item Conclude by showing how the security against quantum storage implies a lower bound on any QFAC on average for $(A,\mathcal{C})$.
 \end{enumerate}
We proceed with the details. 

\begin{proof}
Let $t = \log \overline N$ (assume it an integer for simplicity) and consider the following $1$-bit extractor 
\begin{align*}
E:~&\{0,1\}^N\times \{0,1\}^{t} ~\rightarrow ~\{0,1\}\\
&~~~~~~~(x,y) ~~~~~~~~~\mapsto~  C(x)_y
\end{align*}

The following claim proves item 1 above.  

\begin{claim}\label{claim:one-bit}
$E:\{0,1\}^N \times \{0,1\}^t \rightarrow \{0,1\}$ as defined above is a $(K, \eps)$ strong extractor for any $K > H(\delta)N + \log L + \log \frac{2}{\eps}$.
\end{claim}

\begin{proof}
Assume for the sake of contradiction that $E$ is not a $(K, \eps)$ strong extractor. Then there is a distribution $D$ with min-entropy $K$, and a statistical test $T$ such that the following holds. 
\[ |\Pr_{y \sim U_t, x \sim D} [T(y)=C(x)_y]- \frac{1}{2} | \geq  \eps \]
With a possible flip in the output of circuit $T$, we get a new test $T'$ such that 
\[ \Pr_{y \sim U_t, x \sim D} [T'(y)=C(x)_y]\geq \frac{1}{2}+ \eps \]
By a Markov argument, there is a set $BAD\subseteq  \{0,1\}^N$ such that for every $x \in BAD$, 
\[ \Pr_{y \sim U_t} [T'(y)=C(x)_y]\geq \frac{1}{2}+ \frac{\eps}{2} \]
and $\Pr_{x \sim D} [x \in BAD] \geq \eps/2$. Evaluating $T'$ on every possible $y \in \{0,1\}^t$ results in a string $x'$ such that 
\begin{equation}\label{eq:xor-match}
 \Pr_{y \in \{0,1\}^t} [x'_{y} = C(x)_{y} ] \geq \frac{1}{2} + \frac{\eps}{2}
\end{equation}
We can now use the $(\eps/2,\delta,L)$ list-decodability properties of $C$. For any $x'$  satisfying (\ref{eq:xor-match}) we can get a set of $k\leq L$ strings  $x^1,\ldots,x^k$ such that at least one of them satisfies that 
\begin{equation}\label{eq-close}
\Pr_{y \sim \mathcal{U}_N} [x^i_y=x_y] \geq 1-\delta
\end{equation}
Note that process of finding $x^1,\ldots,x^k$ need not be polynomial time, but we only require existence here; the important point is that the list of $x^i$ is uniquely determined by $x'$ (take the lexicographically smallest list satisfying the conditions in the fact). If $x^1,\ldots,x^k$ are known, then we require at most $\log L$ bits to specify $i \in [t]$ such that $x^i$ satisfies (\ref{eq-close}). Once $x^i$ is specified, we know that $x$ must be among one of the at most $2^{H(\delta) N}$ possible $N$-bit strings which are $\delta$-close to $x$. Hence  we require an additional $H(\delta) N$ bits to fully specify $x$. Thus, the total amount of bits used to specify $x$  is $\log L + H(\delta)N$, which in turn implies that the size of the set $BAD$ is bounded by $L\cdot 2^{H(\delta)N}$. 

To conclude the argument, observe that every element in BAD is sampled with probability at most $2^{-K}$  and hence $\Pr_{X \in D}[X \in BAD] \leq (L\cdot 2^{-K + H(\delta) N})$. However, this is a contradiction if
\[ L\cdot 2^{-K + H(\delta) N} < \frac{\eps}{2} ~~~~~ \textrm{i.e.} ~~~~~ K > H(\delta)N + \log L + \log \frac{2}{\eps} \]
which gives the bound stated in the claim.
\end{proof}

Let $\eta>0$ be an error parameter, $A\subseteq \{0,1\}^N$, and $\mathcal{U}_A$ the uniform distribution on $A$. Theorem III.1 in~\cite{KT07} implies that, as long as
\begin{align}\label{eq:sizeA}
 \log |A| - b \geq K + \log 1/\eta,
\end{align}
the function $E$ is automatically a $(\log |A|, 3\sqrt{\eta})$ extractor that is secure against $b$ qubits of quantum storage (see Definition~\ref{def:strongquantum}). This means that, for any collection of quantum states $\Psi(x)\in \mathcal{D}_b$, knowledge of $y$ and $\Psi(x)$ cannot help distinguish $E(x,y)$ from a uniformly random bit with advantage more than $3\sqrt{\eta}$ (over the choice of $x$ in $A$, and uniform $y$). In particular, we have that for any collection of measurements $\{M_y^0,M_y^1\}_{y\in\{0,1\}^t}$ on $\mathcal{D}_b$, 
$$\Exs{x\in A,\,y\in \{0,1\}^t}{\Tr(M^{C(x)_y}_y \Psi(x)(M^{C(x)_y}_y)^\dagger) } \leq 1/2+3\sqrt{\eta}/2$$
By definition, any $(N,b,\eps)$ QFAC on average for $(A,\mathcal{C})$, even one that is only correct on average over the choice of $x$, contradicts this conclusion for $\eta = 4\eps^2/9$. Hence our assumption \eqref{eq:sizeA} on the size of $A$ must be contradicted, i.e. any such QFAC must be such that $\log |A| < K + b + \log 9/4\eps^2$. Setting $K$ to be the smallest possible value satisfying the condition in Claim~\ref{claim:one-bit}, we get
$$\log |A| < H(\delta)N+b +\log L +O(\log 1/\eps)$$
\end{proof}

We describe two instantiations of this proposition, for specific families of codes. The first one, which will let us get an extractor with optimal seed length, is based on the following from~\cite{GHSZ00}:
\begin{fact}\label{fact:lcode} For any $N\in\N$, $\eps>0$, there exists a polynomial-time computable code $C_R:\{0,1\}^N \rightarrow \{0,1\}^{\overline N}$, where $\overline N = O(N/\eps^4)$, that is $(\eps,O(1/\eps^2))$ list-decodable. 
\end{fact}

These codes lead to the following, the proof of which follows immediately from Proposition~\ref{prop:listqfac}:

\begin{corollary}\label{cor:reedmullerqfac} Let $C_R$ be the code from Fact~\ref{fact:lcode}, and $\mathcal{C}_R = \{f_i:x\mapsto C(x)_i,\, i\in[\overline N]\}$. Then any $(N,b,\eps)$ QFAC on average for $(A,\mathcal{C}_R)$ is such that
$$\log |A| < b + O(\log 1/\eps)$$
Moreover, this bound holds even when we only require the QFAC to have success probability $1/2+\eps$ on average over the choice of $x\in A$. 
\end{corollary}

Our second main construction uses a QFAC for the class $\mathcal{C}_k = \{g:x\mapsto \bigoplus_{j=1}^k  x_{i_j},\ (i_1,\ldots,i_k)\in [n]\}$. QFACs for this class of functions were introduced in~\cite{ARW08}, where they are called XOR-QRACs. That paper shows a bound on the length of such codes using a generalization of the hypercontractive inequality to matrix-valued functions. We improve their result by showing the following:

\begin{corollary}\label{cor:xorqfac} Let $k,N$ be integers, and $\eps>2k^2/2^N$. Let $A\subset \{0,1\}^N$. If there exists a $(N,b,\eps)$ QFAC on average for $(A,\mathcal{C}_k)$, then
$$\log |A| < b+H\left(\frac{1}{k}\ln\frac{2}{\eps}\right) \, N + O\left(\log \frac{1}{\eps}\right) $$
Moreover, this bound holds even when we only require the QFAC to have success probability $1/2+\eps$ on average over the choice of $x\in A$. 
\end{corollary}

By generalizing the proof of Theorem~7 in~\cite{ARW08} (which is only stated for $A = \{0,1\}^N$ in that paper), we can get the bound $\log |A| \leq b+\left(1-\frac{1}{2\ln2}(2\eps)^{2/k}+o_N(1)\right)N $ for all $k\geq \log\log N$. This would lead to an extractor construction which only works for sources with min-entropy $\gamma N$ for $\gamma > 0.28$, and our improvement on their bound gets rid of this constraint. 

\begin{proof}
The following lemma (for a reference, see \cite{IJK06}, Lemma 42) shows that for any $\eps>2k^2/2^N$, the XOR code is $(\eps,(1/k) \ln (2/\eps),4/\eps^2)$ approximately list-decodable.

\begin{lemma}
For every $\eps>2k^2/2^N$ and $z' \in (\{0,1\}^N)^k$, there is a list of $t \leq 4/\eps^2$ elements $x^1,\ldots,x^t \in \{0,1\}^N$ such that the following holds:
for every $z \in \{0,1\}^N$ which satisfies 
$$  \Pr_{\{y_1,\ldots,y_k\} \in \binom{N}{k}} [z'_{(y_1,\ldots,y_k)} = \oplus_{i=1}^k z_{y_i} ] \geq \frac{1}{2} + \eps  $$
  there is an $i \in [t]$ such that $$\Pr_{y \sim \mathcal{U}_N} [x^i_y = z_y] \geq 1-\delta$$
with $\delta = (1/k) \ln (2/\eps)$. 
\end{lemma}
Note that in~\cite{IJK06} the lemma is proved for tuples instead of sets, and has a $t\leq 1/\eps^2$. However, since most tuples are sets, it is straightforward to get the above version for sets. Plugging the list-decoding parameters from this lemma in the bound of Proposition~\ref{prop:listqfac} immediately gives the result.
\end{proof}

\section{Overview of the construction}\label{sec:main}

Our construction follows the general paradigm introduced by Trevisan~\cite{T99} and its subsequent adaptation against quantum storage by Ta-shma~\cite{TS09}. However, our proof technique differs from that of~\cite{TS09} in that it avoids constructing random access codes by copying the adversary's storage many times. Rather, we use the much stronger bounds on QFACs proved in Section~\ref{sec:qfac}. This is crucial in allowing us to prove an additive, rather than multiplicative, dependence of the output on the adversary's storage.

We first describe a few standard tools that are used in the construction, before giving it in detail. Its correctness will be proved in Section~\ref{sec:mainsecurity}.

\subsection{Preliminaries}

 \begin{definition}
A collection of subsets $S_1,\ldots,S_m \subset [t]$ is called a $(t,n,m,\rho)$ weak design if for all $i$, $|S_i|=n$ and for all $j$, $\sum_{i<j} 2^{|S_i \cap S_j|} \leq \rho(m-1)$.
\end{definition}

The following theorem is due to Raz, Reingold and Vadhan \cite{RRV99}.

\begin{theorem}\label{thm:1}
For every $m,t \in \mathbb{N}$ and $\rho\geq 2$, there is a $(t,n,m,\rho)$ design which is computable in time $(mt)^{O(1)}$ with $t =O(\frac{n^2}{\log \rho})$.
\end{theorem}

Note that the value of $t$ blows up when $\rho$ approaches $1$. In order to keep $t$ bounded even as $\rho$ approaches $1$, we can use a construction given in \cite{RRV99}.  Even though the construction is  computable in polynomial time, it does not meet many finer notions of efficiency which are of interest to us. Hartman and Raz \cite{HR01} achieved similar parameters with a better efficiency:

\begin{theorem}\label{thm:12}
For every $m,t \in \mathbb{N}$ such that $m>n^{\log n}$ and $0<\gamma<\frac{1}{2}$, there is a $(t,n,m,1+\gamma)$ design such that $t =O(n^2 \log \frac{1}{\gamma})$. Further, each individual set in the design can be output in time polynomial in $t$ and $n$. 
\end{theorem}

For the purposes of this paper, let $t(n,\rho)$ denote the smallest value of $t$  for which Theorems~\ref{thm:1} or~\ref{thm:12} guarantee the existence of a weak $(t,n,m,\rho)$ design. Whether we use Theorem~\ref{thm:1} or \ref{thm:12} depends on how small we want $\rho$ to be. 

Our last tool is the Nisan-Wigderson generator with respect to a function $f:\{0,1\}^n \rightarrow \{0,1\}$. 

\begin{definition}
Let $S_1,\ldots,S_m$ be a $(t,n,m,\rho)$ weak-design. Let $x\in\{0,1\}^{2^n}$. Then $NW^x:\{0,1\}^t \rightarrow \{0,1\}^m$ is defined as
\[ NW^x(y) = x_{y_{S_1}}, \ldots, x_{y_{S_m}} \]
Here $x_{S_j}$ denotes the restriction of $x$ to the indices in $S_j$.
\end{definition}

\subsection{Description of the construction}\label{sec:description}

Let $C:\{0,1\}^N\to\{0,1\}^{\overline N}$ be a code with good (possibly approximate) list-decoding capabilities, and $(S_1,\ldots,S_m)$ be a $(t,\log \overline N,m,\rho)$ design as discussed above. Then the extractor is obtained by combining these two constructs in the following way: 
\begin{align*}
Ext_C:\,&\{0,1\}^{N} \times \{0,1\}^t \rightarrow \{0,1\}^m\\
& (x,y) \mapsto  NW^{C(x)}(y)
\end{align*}


\subsection{Main theorem}

Our main result is the following:

\begin{theorem}\label{thm:XORextmain}
Let $\delta,\eps>0$. Let $C:\{0,1\}^N\to\{0,1\}^{\overline N}$ be a $(\eps/m,\delta,L)$ approximately list-decodable code, and $t = t(\log \overline{N} ,\rho)$ such that there exists a $(t,\log\overline N,m,\rho)$ design for all large enough $m$.  Then for any $K,b\in\N$ the function $Ext_C:\{0,1\}^N \times \{0,1\}^t \rightarrow \{0,1\}^m$ is a $(K, 2\eps)$ extractor secure against $b$ qubits of quantum storage, where 
$$ m=\frac{K-b-t-H(\delta)N -\log L-\Omega(\log(1/\eps) + \log N)}{1+\rho}$$ 
\end{theorem}

We give two instantiations of this result. The first one uses the codes from Fact~\ref{fact:lcode}, and lets us achieve optimal seed length. We obtain it by setting $\rho=K^{\gamma/2}$, for any $\gamma>0$, and using the combinatorial designs guaranteed by Theorem~\ref{thm:1}:

\begin{corollary}\label{cor:lcode} Let $\gamma,c,c'>0$ be any constants. Let $C_R$ be the code obtained from Fact~\ref{fact:lcode} by setting $\eps=N^{-c}$. Then the function $Ext_{C_R}:\{0,1\}^{N} \times \{0,1\}^t \rightarrow \{0,1\}^m$, where $t = O(\log N)$ and $m=\Omega\left(\frac{K-b}{K^\gamma}\right)$, is a $(K, 2\eps)$ extractor against $b$ qubits of quantum storage for any $K\geq N^{c'}$. 
\end{corollary}

An inconvenient aspect of this construction, particularly relevant to cryptography, is that, even though the extractor is polynomial-time computable, it is not \emph{locally computable}. Indeed, any bit of the output may require polynomial time to be computed, whereas one might wish for it to be computable in polylogarithmic time. We achieve such an extractor by taking $C = C_k:\{0,1\}^N \to \{0,1\}^{\binom{N}{k}}$ the XOR code $C_k(x)_{y_1,\ldots,y_k} = x_{y_1}\oplus\ldots \oplus x_{y_k}$. By using these codes together with the designs from Theorem~\ref{thm:12}, the bound from Corollary~\ref{cor:xorqfac} gives the following:

\begin{corollary}\label{thm:corr1}
Let $\alpha,\delta,c>0$  be any constants. Then there is a $k = O(\log(m/\eps)/\delta^2)$ such that the function $Ext_{C_k}:\{0,1\}^{N} \times \{0,1\}^t \rightarrow \{0,1\}^m$,  where $t = O(\log^4 N)$ and $m=\frac{1}{2}((\alpha -2 \delta)N - b )$, is a $(\alpha N, N^{-c})$ extractor against $b$ qubits of quantum storage.
\end{corollary}

Note that this extractor is locally computable, and every individual bit of the output can be computed in polylogarithmic time, as the designs in Theorem~\ref{thm:12} are locally computable. Note also that the extractor only works for linear entropy rates: as mentioned earlier, this is tight due to lower bounds by Viola~\cite{V03} on the seed length required to extract from sources with polynomially small min-entropy using low complexity circuits.

\section{Proof of security}\label{sec:mainsecurity}

We give the proof of security of our construction. The first steps of the proof follow the general reconstruction paradigm from~\cite{T99}, and we give them first. 

\subsection{Proofs in the reconstruction paradigm}\label{sec:reconstruction}

We start with the following standard observation.

\begin{observation}\label{use2}
In order to prove that $Ext:\{0,1\}^N\times\{0,1\}^t\rightarrow\{0,1\}^m$ is a $(K,2\eps)$ strong extractor against $b$ qubits of storage, it suffices to prove that for any collection of measurements $\{M_{u,y}^1,M_{u,y}^0\}_{(u,y)\in\{0,1\}^{m+t}}$ on $\mathcal{D}_b$, and $\Psi:\{0,1\}^N \rightarrow \mcal{D}_b$, there are at most $\eps 2^{K}$ strings $x\in\{0,1\}^N$  such that 
\begin{align}\label{an2}
\abs{\Exs{y\in\{0,1\}^t}{\Tr\Big(\Exs{u\in\{0,1\}^m}{ M_{u,y}^1 \Psi(x)(M_{u,y}^1)^\dagger}\Big)  - \Tr\Big(M_{Ext(x,y),y}^1\Psi(x)(M_{Ext(x,y),y}^1)^\dagger\Big)}}  > \eps 
\end{align}
\end{observation}
 
\begin{proof}
Assume for contradiction that $Ext:\{0,1\}^N \times \{0,1\}^t \rightarrow \{0,1\}^m$ is not a $(K,2\eps)$ strong extractor against $b$ qubits of quantum storage. By definition, there exist measurements\\ $\{M_{u,y}^1,M_{u,y}^0\}_{(u,y)\in\{0,1\}^{m+t}}$ on $b$ qubits such that    
$$\abs{\Exs{x\sim X, y\sim U_t}{\Tr\Big(\Exs{u\in\{0,1\}^m}{ M_{u,y}^1 \Psi(x)(M_{u,y}^1)^\dagger}\Big)  - \Tr\Big(M_{Ext(x,y),y}^1\Psi(x)(M_{Ext(x,y),y}^1)^\dagger\Big)}} > 2\eps $$
where $X$ is the source's distribution. Since it has min-entropy at least $K$, it must be true that for at least $\eps 2^{K}$ inputs $x$,
$$\abs{\Exs{y\sim U_t}{\Tr\Big(\Exs{u\in\{0,1\}^m}{ M_{u,y}^1 \Psi(x)(M_{u,y}^1)^\dagger}\Big)  - \Tr\Big(M_{Ext(x,y),y}^1\Psi(x)(M_{Ext(x,y),y}^1)^\dagger\Big)}}  > \eps $$
\end{proof}

Fix a collection of measurements $\{M_{u,y}^1,M_{u,y}^0\}_{u,y\in\{0,1\}^{m+t}}$ on $\mathcal{D}_b$. The previous observation shows that, in order to show that $Ext$ is a strong extractor, it suffices to bound the number of strings $x$ such that~(\ref{an2}) holds. For this, we use the reconstruction approach  in \cite{T99}. For a fixed $x$, define $M_x:\{0,1\}^{m+t}\rightarrow\{0,1\}$ as the probabilistic procedure which, on input $(u,y)\in \{0,1\}^{m+t}$, outputs $1$ with probability $\Tr(M^1_{u,y} \Psi(x)(M^1_{u,y})^\dagger)$, where $\Psi(x)$ is the state of the adversary's storage on $x$.
For the most part our proofs will simply treat $M_x$ as a probabilistic oracle. Moreover, all probabilities that we write involving $M_x$, or other oracle circuits making calls to $M_x$, will implicitly be taken over $M_x$'s internal randomness. 

The first step is to use the  standard hybrid argument followed by Yao's distinguisher versus predictor lemma to get an oracle circuit $T$ which queries $M_x$ exactly once, and is such that $T$ predicts $Ext(x,y)_i$ with some advantage over a random guess when $y$ as well as the value of $x$ on some related points are given as input. We skip the (by now, standard) argument and state the final result (see \cite{T99} for details).

\begin{lemma}\label{use1}
Let $x,\eps$ be such that (\ref{an2}) is satisfied, and $Ext(x,y)_i$ be the $i^{th}$ bit of the extractor's output on $(x,y)$. Then  using $ m+\log m+3$ bits of classical advice, we can construct an oracle circuit $T$ which makes one query to $M_x$ and is such that for some $1 \leq i \leq m$,  $T$ satisfies:
\begin{equation}\label{an4}
Pr_{y \in U_t}[T^{M_x}(y,Ext(x,y)_1,\ldots,Ext(x,y)_{i-1})=Ext(x,y)_{i}] \geq \frac{1}{2}+\frac{\eps}{m}
\end{equation}
\end{lemma}

Our next step is to construct a small circuit $R_x$ which predicts the value of $C(x)$ at any position $y$ with some non-negligible success probability, leading to the following technical lemma:

\begin{lemma}\label{lem:main1}
Let $x,\eps$ be such that (\ref{an2}) is satisfied. Then using $m(1+\rho) +\log m + t +O(1)$ bits of classical advice, we can construct an oracle circuit $R_x$ which makes one query to $M_x$ and predicts $C(x)_z$ with probability $1/2+\eps/m$, on average over the choice of $z\in\{0,1\}^{\overline{N}}$.
\end{lemma}

\begin{proof}
By Lemma~\ref{use1}, using $m+\log m + 3$ bits of advice, we can get an oracle circuit $T$ which makes exactly one query to $M_x$ and for some $1\leq i \leq m$ satisfies 
\[ \Pr_{y} [T^{M_x}(y,C(x)_{y_{S_1}},\ldots,C(x)_{y_{S_{i-1}}})=C(x)_{y_{S_i}}] \geq \frac{1}{2} + \frac{\eps}{m} \]
Let us split $y$ into two parts $z=y_{S_i}$ and $w=y_{[t]-S_i}$. Let $y_{S_j}$ be denoted by $h_j(z,w)$. The above probability can then be rewritten as 
\[ \Pr_{z, w} [T^{M_x}(z,w,C(x)_{h_1(z,w)},\ldots,C(x)_{h_{i-1}(z,w)})=C(x)_z] \geq \frac{1}{2} + \frac{\eps}{m} \]
By an averaging argument, we can fix a $w$ (using at most $t$ bits of advice) such that the above inequality holds with the probability taken over $z$. Let us hardwire all the possible values of $C(x)_{h_j(z,w)}$ (for the fixed value of $w$), as $z$ varies over $\{0,1\}^{\overline{N}}$ and $j$ varies between $1$ and $i-1$, into the circuit $T$. By the definition of a weak design, there are at most $(m-1)\rho $ bits that need to be hardwired. Let $R_x$ be the circuit with all the hardwired values. $R_x$ satisfies the following
\begin{equation}\label{eqn:xorlem}
\Pr_{z} [R_x^{M_x}(z)=C(x)_z] \geq \frac{1}{2} + \frac{\eps}{m} 
\end{equation}
The total classical advice taken so far is $m + \log m + t + m \rho + O(1)$. 
\end{proof}

\subsection{Security against quantum storage from lower bounds on QFACs}\label{sec:qfacsecurity}

Assume for contradiction that there is an adversary to $Ext$, which can distinguish its output from uniform given access to the seed $y$ and some partial quantum information $\Psi(x)\in\mathcal{D}_b$ about the source. Such an adversary can be described by the mapping $\Psi$, together with a collection of measurements $\{M_{u,y}^1,M_{u,y}^0\}_{u,y\in\{0,1\}^{m+t}}$ on $\mathcal{D}_b$ describing the adversary's measurement on his quantum information $\Psi(x)$, when provided with the seed and the extractor's output\footnote{This describes the most general situation, as we can always assume that any measurement made by the adversary is done at the end of his recovery procedure.}. 

For a fixed $x$, let $M_x:\{0,1\}^{m+t}\rightarrow\{0,1\}$ as in Section~\ref{sec:reconstruction}. By Observation~\ref{use2}, to prove that $Ext$ is a $(K,2\eps)$ strong extractor secure against $b$ qubits of quantum storage, it suffices to prove that there are at most $\eps2^{K}$ strings $x$ such that~\eqref{an2} holds.

The key conceptual step in our proof is to observe that from the circuit $R_x$ given by Lemma~\ref{lem:main1}, we  can construct a QFAC for the family $\mathcal{C} = \{f_i:x\mapsto C(x)_i,\, i\in [\overline{N}]\}$ of codeword positions, and the set $A$ of all $x$ satisfying~\eqref{an2}. The strong lower bounds we proved in Section~\ref{sec:qfac} then let us bound the size of the set $A$ as a function of the adversary's storage and the list-decoding properties of $C$. The following claim makes this connection formal.

\begin{claim}\label{obs2}
Let $\eta>0$ and $A\subseteq\{0,1\}^N$ be such that, for any $x\in A$, using only $c$ bits of classical advice, we can construct a circuit $R_x$ which has access to a $b$-qubit quantum state and is such that for a random $y$, it predicts $C(x)_y$ with probability $1/2+\eta$. Then the cardinality of $A$ is at most $s\cdot 2^c$, where $s$ is the maximum size of a set $B$ such that there exists a $(N,b,\eta)$ QFAC for $(B, \mathcal{C})$.
\end{claim}

\begin{proof}
The  $c$ advice bits partition the set $A$ into $2^c$ sets $A_s$, for $s\in\{0,1\}^c$. Fix such a $s$ and consider the set $A_s$. Since $s$ has been fixed, all $x\in A_s$ have the same circuit $R_x$; only the $b$-qubit quantum state $\Psi(x)$ on which it operates depends on $x$. Hence there is a fixed set of measurements such that, for a random $y$, the measurement $M_y$ on $\Psi(x)$ outputs $C(x)_y$ with probability $1/2+\eta$. This means we have a $(N,b,\eta)$ QFAC for $(A_s,\mathcal{C})$. Hence the size of $A_s$ is bounded by the maximum size of any set for which such a code exists. This gives us the promised bound on $A$.
\end{proof}  

To finish the proof of Theorem~\ref{thm:XORextmain}, note that by Proposition~\ref{prop:listqfac}, any $(N,b,\eps/m)$ QFAC for $(A,\mathcal{C})$ satisfies
 $$ \log |A| \leq b+ H(\delta)\, N +\log L + O(\log m/\eps)  $$

Applying Claim~\ref{obs2} to the advice circuit promised by Lemma~\ref{lem:main1}, we deduce that the number of strings $x$ such that (\ref{an2}) holds is at most $2^{b+H(\delta)N +\log L+O(\log(m/\eps)) }\cdot 2^{m(1+\rho) +\log m + t + O(1)}$. 
Using $\log(m) = O(\log N)$, this expression can be upper-bounded by 
$$2^{b+H(\delta)N + m(1+\rho) +\log L + t+ O(\log (1/\eps) + \log  N)}$$
 Using the bound on $m$ given in Theorem~\ref{thm:XORextmain}, we immediately get that this expression is upper-bounded by $\eps2^K$, finishing the proof of the theorem.

\vskip1cm

\paragraph{Acknowledgements.} We are grateful to Falk Unger and Umesh Vazirani for helpful comments concerning the presentation of this manuscript. The first author would like to thank Luca Trevisan for kindly sharing his understanding of extractors.
\bibliography{macros,luca}

\newcommand{\etalchar}[1]{$^{#1}$}
\begin{thebibliography}{BARdW08}

\bibitem[ANTsV02]{ANTV02}
Andris Ambainis, Ashwin Nayak, Amnon Ta-shma, and Umesh~V. Vazirani.
\newblock Dense quantum coding and quantum finite automata.
\newblock {\em Journal of the ACM}, 49(4):496--511, 2002.

\bibitem[BARdW08]{ARW08}
Avraham Ben-Aroya, Oded Regev, and Ronald de~Wolf.
\newblock {A Hypercontractive Inequality for Matrix-Valued Functions with
  Applications to Quantum Computing and LDCs }.
\newblock In {\em Proceedings of the 49th IEEE Symposium on Foundations of
  Computer Science}, pages 477--486, 2008.
\newblock Full version at arXiv:0705.3806.

\bibitem[CDNT98]{CDNT98}
Richard Cleve, Wim~van Dam, Michael Nielsen, and Alain Tapp.
\newblock Quantum entanglement and the communication complexity of the inner
  product function.
\newblock In {\em QCQC '98: Selected papers from the First NASA International
  Conference on Quantum Computing and Quantum Communications}, pages 61--74,
  London, UK, 1998. Springer-Verlag.

\bibitem[DM04]{DM04}
Stefan Dziembowski and Ueli Maurer.
\newblock Optimal randomizer efficiency in the bounded-storage model.
\newblock {\em Journal of Cryptology}, 17(1):5--26, 2004.

\bibitem[DPRV09]{DPRV09}
Anindya De, Christopher Portmann, Renato Renner, and Thomas Vidick.
\newblock Trevisan's extractor in the presence of quantum side information.
\newblock Technical report arXiv:0912.5514, 2009.

\bibitem[DT09]{DT09}
Anindya De and Luca Trevisan.
\newblock Extractors using hardness amplification.
\newblock In {\em APPROX-RANDOM}, pages 462--475, 2009.
\newblock Full version available at \emph{
  http://www.cs.berkeley.edu/$\sim$anindya/exthardfull.pdf}.

\bibitem[FS08]{FS08}
Serge Fehr and Christian Schaffner.
\newblock Randomness extraction via {\it delta} -biased masking in the presence
  of a quantum attacker.
\newblock In Ran Canetti, editor, {\em TCC}, volume 4948 of {\em Lecture Notes
  in Computer Science}, pages 465--481. Springer, 2008.

\bibitem[GHSZ02]{GHSZ00}
Venkatesan Guruswami, Johan H{\aa}stad, Madhu Sudan, and David Zuckerman.
\newblock Combinatorial bounds for list decoding.
\newblock {\em IEEE Transactions on Information Theory}, 48(5):1021--1034,
  2002.

\bibitem[GKK{\etalchar{+}}07]{GKKRW}
Dmitri Gavinsky, Julia Kempe, Iordanis Kerendis, Ran Raz, and Ronald de~Wolf.
\newblock Exponential separations for one-way quantum communication complexity
  with applications to cryptography.
\newblock In {\em Proceedings of the 39th ACM Symposium on Theory of
  Computing}, pages 516--525, 2007.

\bibitem[Hol73]{Hol73}
Alexander Holevo.
\newblock Information-theoretic aspects of quantum measurement.
\newblock {\em Problems of Information Transmission}, 9(2):31--42, 1973.

\bibitem[HR03]{HR01}
Tzvika Hartman and Ran Raz.
\newblock On the distribution of the number of roots of polynomials and
  explicit weak designs.
\newblock {\em Random Structures and Algorithms}, 23(3):235--263, 2003.

\bibitem[IJK06]{IJK06}
Russell Impagliazzo, Ragesh Jaiswal, and Valentine Kabanets.
\newblock {Approximately List-Decoding Direct Product Codes and Uniform
  Hardness Amplification}.
\newblock In {\em Proceedings of the 47th IEEE Symposium on Foundations of
  Computer Science}, pages 187--196, 2006.
\newblock Full version at {\tt http://www1.cs.columbia.edu/$\sim$rjaiswal/}.

\bibitem[KMR05]{KMR05}
Robert K\"onig, Ueli Maurer, and Renato Renner.
\newblock On the power of quantum memory.
\newblock {\em IEEE Transactions on Information Theory}, 51(7):2391--2401,
  2005.

\bibitem[KT08]{KT07}
Robert K\"onig and Barbara Terhal.
\newblock The bounded storage model in presence of a quantum adversary.
\newblock {\em IEEE Transactions on Information Theory}, 54(2):749--762, 2008.

\bibitem[Lu04]{L04}
Chi-Jen Lu.
\newblock Encryption against storage-bounded adversaries from on-line strong
  extractors.
\newblock {\em Journal of Cryptology}, 17(1):27--42, 2004.

\bibitem[Mau92]{Maurer92}
Ueli~M. Maurer.
\newblock Conditionally-perfect secrecy and a provably-secure randomized
  cipher.
\newblock {\em Journal of Cryptology}, 5(1):53--66, 1992.

\bibitem[NS06]{NS06}
Ashwin Nayak and Julia Salzman.
\newblock Limits on the ability of quantum states to convey classical messages.
\newblock {\em Journal of the ACM}, 53(1):184--206, 2006.

\bibitem[NW94]{NW94}
Noam Nisan and Avi Wigderson.
\newblock Hardness vs randomness.
\newblock {\em Journal of Computer and System Sciences}, 49:149--167, 1994.
\newblock Preliminary version in {\em Proc. of FOCS'88}.

\bibitem[RRV99]{RRV99}
R.~Raz, O.~Reingold, and S.~Vadhan.
\newblock Extracting all the randomness and reducing the error in {Trevisan}'s
  extractors.
\newblock In {\em Proceedings of the 31st ACM Symposium on Theory of
  Computing}, pages 149--158, 1999.

\bibitem[STV01]{STV99}
Madhu Sudan, Luca Trevisan, and Salil Vadhan.
\newblock Pseudorandom generators without the {XOR} lemma.
\newblock {\em Journal of Computer and System Sciences}, 62(2):236--266, 2001.
\newblock Preliminary version in STOC-Complexity 99.

\bibitem[Tre01]{T99}
Luca Trevisan.
\newblock Extractors and pseudorandom generators.
\newblock {\em Journal of the ACM}, 48(4):860--879, 2001.

\bibitem[Ts09]{TS09}
Amnon Ta-shma.
\newblock Short seed extractors against quantum storage.
\newblock In {\em Proceedings of the 41st ACM Symposium on Theory of
  Computing}, pages 401--409, 2009.

\bibitem[Vad04]{V04}
Salil~P. Vadhan.
\newblock Constructing locally computable extractors and cryptosystems in the
  bounded-storage model.
\newblock {\em Journal of Cryptology}, 17(1):43--77, 2004.

\bibitem[Vio04]{V03}
Emanuele Viola.
\newblock The complexity of constructing pseudorandom generators from hard
  functions.
\newblock {\em Computational Complexity}, 13(3-4):147--188, 2004.

\end{thebibliography}
\bibliographystyle{alpha}

\end{document}